\documentclass[journal,onecolumn]{IEEEtran}

\usepackage[T1]{fontenc}
\usepackage{graphicx,colordvi,psfrag}
\usepackage{amsmath,amssymb}
\usepackage{calc,pstricks, pgf, xcolor}
\usepackage{bbding}
\usepackage{bbm}
\usepackage{dsfont}
\usepackage{stfloats}

\newcommand{\RR}{\mathbb{R}}

\def\eqdef{\triangleq}

\newcommand{\m}{\mathcal}

\def\calX{\mathcal{X}}
\def\calY{\mathcal{Y}}

\newtheorem{theorem}{Theorem}

\newenvironment{proof}[1][Proof]{\noindent\textbf{#1.} }{\ \rule{0.5em}{0.5em}}

\def\etaKL{\eta_{\mathrm{KL}}}
\def\eqdef{\triangleq}

\begin{document}
\title{Strong Data Processing Constant is Achieved by Binary Inputs}

\author{Or Ordentlich 
and	Yury Polyanskiy 
\thanks{O. Ordentlich is with the 
Hebrew University of Jerusalem, Israel (\texttt{or.ordentlich@mail.huji.ac.il}).  Y. Polyanskiy is with the MIT,
USA (\texttt{yp@mit.edu}).}
}

\date{}

\maketitle
\begin{abstract}
For any channel $P_{Y|X}$ the strong data processing constant is defined as the smallest number $\eta_{KL}\in[0,1]$ such that
$I(U;Y)\le \eta_{KL} I(U;X)$ holds for any Markov chain $U-X-Y$. It is shown that the value of $\eta_{KL}$ is given by
that of the best binary-input subchannel of $P_{Y|X}$. The same result holds for any $f$-divergence, verifying a conjecture of Cohen, Kemperman and Zbaganu (1998).
\end{abstract}

Consider an arbitrary channel $P_{Y|X}:\calX\to \calY$ with countable $\calX$. We define the strong data processing
inequality (SDPI) constant~\cite{ahlswede1976spreading}
\begin{equation}\label{eq:etakl_def}
	\etaKL = \sup {D(P_{Y|X} \circ P\| P_{Y|X}\circ Q) \over D(P\|Q)}\,, 
\end{equation}
where optimization is over all pairs  of distributions on $\calX$, denoted $P,Q\in \mathcal{P}(\calX)$, such that
$0<D(P\|Q)<\infty$, and $P_{Y|X}\circ P$ is the distribution of the output $Y$ when the input $X$ is distributed according to $P \in
\mathcal{P}(\calX)$. We refer to~\cite{pw17} for a
survey of the properties and importance of the SDPI, in particular for showing equivalence to the definition in the
abstract, and advertise~\cite{polyanskiy2020application} as a recent application in statistics.


When the input alphabet $\calX$ is binary, the value of $\etaKL$ is relatively easy to compute, cf. ~\cite[Appendix
B]{pw17}
. Here we prove that for general $\calX$ determination of $\etaKL$ can be reduced to the binary case.
\begin{theorem} Optimization in~\eqref{eq:etakl_def} can be restricted to pairs $P,Q$ supported on two
points in $\calX$ (same for both).
\label{thm:main}
\end{theorem}
\begin{proof}
For two distributions $P$ and $Q$ on $\m{X}$ and $\lambda\in(0,1)$ define
\vspace{-2mm}
\begin{align}
L_\lambda(P,Q)&\triangleq D(P_{Y|X}\circ P\| P_{Y|X}\circ Q)-\lambda D(P\|Q).\nonumber
\end{align}
\vspace{-1mm}
We assume that $0<D(P\|Q)<\infty$ as required by the definition of $\etaKL$. 
We will show that we can find two distributions $\hat{P}$ and
$\hat{Q}$ where $\hat{Q}$ is supported on two letters in $\mathrm{supp}(Q)\triangleq \{x\in\m{X} \ : \ Q(x)>0 \}$, and
$L_\lambda(\hat{P},\hat{Q})\geq L_\lambda(P,Q)$. This implies the statement, since
$\etaKL=\sup\left\{\lambda \  : \ \sup_{P,Q} L_\lambda(P,Q)\geq 0 \right\}$.


To that end define the convex set of distributions
\begin{align}
\m{S}\triangleq \bigg\{\hat{Q}: &\mathrm{supp}(\hat{Q})\subseteq \mathrm{supp}(Q),\nonumber\\ 
&\sum_{x\in  \mathrm{supp}(Q)} \frac{P(x)}{Q(x)}\cdot \hat{Q}(x)=1 \bigg\}.\nonumber
\end{align}
Consider the function $g:\m{S}\to\RR$ defined as $g(\hat{Q})=L_\lambda\left(\frac{P}{Q}\hat{Q},\hat{Q}\right)$. Note that $Q\in\m{S}$ and $g(Q)=L_{\lambda}(P,Q)$. Consequently, 
$\max_{\hat{Q}\in\m{S}} g(\hat{Q})\geq L_{\lambda}(P,Q)$.
Note that
\begin{align}
\hat{Q}\mapsto D\left(P_{Y|X}\circ \frac{P}{Q}\hat{Q}\bigg\| P_{Y|X}\circ \hat{Q}\right)\nonumber
\end{align}
is convex by convexity of $(P,Q)\mapsto D(P\|Q)$, and that
\begin{align}
\hat{Q}\mapsto D\left(\frac{P}{Q}\hat{Q}\bigg\| \hat{Q}\right)=\sum_x \hat Q(x) {P(x)\over Q(x)}\log {P(x)\over
Q(x)}\nonumber
\end{align}
is linear. Thus, $\hat{Q} \mapsto g(\hat Q)$ is convex on $\m{S}$. It therefore follows that $\max_{\hat{Q}\in\m{S}} g(\hat{Q})$ is obtained at an extreme point of $\m{S}$. Since $\m{S}$ is the intersection of the simplex with a hyperplane, its extreme points are supported on at most two atoms. 
\end{proof}

Paired with~\cite[Appendix B]{pw17} we get a corollary bounding $\etaKL$ in terms of the
Hellinger-diameter of the channel:
	\begin{align}
{1\over 2} \mathrm{diam_{H^2}}(P_{Y|X}) \le \etaKL &\le g\left({1\over 2}
\mathrm{diam_{H^2}}(P_{Y|X})\right)\nonumber\\
		&\le \mathrm{diam_{H^2}}(P_{Y|X})
	\end{align}
	where $g(t)\eqdef 2t\left(1-\frac{t}{2}\right)$, $\mathrm{diam_{H^2}}(P_{Y|X}) =
	\sup_{x,x'} H^2(P_{Y|X=x}, P_{Y|X=x'})$ and $H^2(P,Q)=2-2\int\sqrt{dP dQ}$.

Note that the only property of divergence that we have used in the proof of Theorem~\ref{thm:main} is convexity of $(P,Q)\mapsto D(P,Q)$. This property is shared
by all $f$-divergences, cf.~\cite{csiszar67}. In other words we proved:
\begin{theorem} Let $\eta_f = \sup {D_f(P_{Y|X} \circ P \| P_{Y|X} \circ Q) \over D_f(P\|Q)}$ optimized over all $P,Q
\in \mathcal{P}(\calX)$ with $0<D_f(P,Q)<\infty$. Then the optimization can be restricted to pairs $P,Q$ supported on two
common points in $\calX$. 
\end{theorem}
This fact was conjectured in~\cite[Open Problem 7.4]{cohen1998comparisons}.


There are two other noteworthy results that our technique entails. First,
a moment of reflection confirms that we, in fact, have shown that
the upper concave envelope of the set $\cup_{P_X,Q_X} \{(D_f(P_X\|Q_X),
D_f(P_Y\|Q_Y))\}$ is unchanged if we restrict the union to pairs $P_X,Q_X$ supported on two
points.

Second, a similar argument holds for the post-SDPI coefficient of a
channel~\cite{YP2019-postsdpi}, defined as 
$$ \eta^{(p)}_{KL}(P_{Y|X}) = \inf\{\eta:
I(U;X) \le \eta I(U;Y) \quad \forall X-Y-U\}\,.$$ 
Namely, we have that $\eta^{(p)}_{KL}$ can be
computed by restricting $X$ to take two values. Indeed, fix an arbitrary $P_{X,Y,U}$ s.t. $X-Y-U$. 
As shown in~\cite[Theorem 4]{pw17} one can
safely assume $U$ to be binary. Now, consider a set $\m{S}$ of all $\hat P_X$ such that the joint
distribution $\hat P_{X,Y,U} = \hat P_X P_{Y|X}
P_{U|Y}$ satisfies $\hat P_U=P_U$. Since $U$ is binary, $\m{S}$ is an intersection of a hyperplane with
a simplex. Now, the function $\hat P_X \mapsto \hat I(U;X)-\lambda \hat
I(U;Y)$ is linear in $\hat P_X$ over $\m{S}$. Consequently, the maximum (and the minimum) of this
function is attained at a binary $\hat P_X$.

\end{document}